\documentclass[reqno,12pt]{amsart}

\usepackage{amsmath}

\newtheorem{theorem}{Theorem}[section]
\newtheorem{definition}{Definition}[section]
\newtheorem{proposition}[theorem]{Proposition}

\newtheorem{remark}[theorem]{Remark}
\newtheorem{hypothesis}[theorem]{Hypothesis {\bf H.}\hspace*{-0.6ex}}

\newcommand{\R}{{\mathbb R}}

\newcommand{\C}{{\mathbb C}}

\newcommand{\nn}{\nonumber}
\newcommand{\be}{\begin{equation}}
\newcommand{\ee}{\end{equation}}
\newcommand{\bea}{\begin{eqnarray}}
\newcommand{\eea}{\end{eqnarray}}
\newcommand{\ba}{\begin{array}}
\newcommand{\ea}{\end{array}}
\newcommand{\ul}{\underline}
\newcommand{\ol}{\overline}

\newcommand{\id}{\mathbb{I}}

\newcommand{\re}{\mathrm{Re}}
\newcommand{\im}{\mathrm{Im}}

\newcommand{\ca}{\mathcal{A}}
\newcommand{\cb}{\mathcal{B}}
\newcommand{\cK}{\mathcal{K}}

\newcommand{\sig}{\sigma}
\newcommand{\lam}{\lambda}
\newcommand{\Lam}{\Lambda}

\newcommand{\om}{\omega}
\newcommand{\Om}{\Omega}
\newcommand{\x}{\xi}
\newcommand{\dta}{\delta}
\newcommand{\Dta}{\Delta}
\newcommand{\al}{\alpha}
\newcommand{\tha}{\theta}

\numberwithin{equation}{section}

\begin{document}
\title[DNLS equation on the interval]{The derivative nonlinear Schr\"odinger equation on the interval}

\author[J.Xu]{Jian Xu}
\address{School of Mathematical Sciences\\
Fudan University\\
Shanghai 200433\\
People's  Republic of China}
\email{11110180024@fudan.edu.cn}

\author[E.Fan]{Engui Fan$^1$}\footnote{Corresponding author and e-mail: faneg@fudan.edu.cn}
\address{School of Mathematical Sciences, Institute of Mathematics and Key Laboratory of Mathematics for Nonlinear Science\\
Fudan University\\
Shanghai 200433\\
People's  Republic of China}
\email{faneg@fudan.edu.cn}

\keywords{Riemann-Hilbert problem, DNLS equation, global relation, finite interval, initial-boundary value problem}

\date{\today}

\begin{abstract}
We use the Fokas method to analyze the derivative nonlinear Schr\"odinger (DNLS) equation $iq_t(x,t)=-q_{xx}(x,t)+(r q^2)_x$ on the interval $[0,L]$. Assuming that the solution $q(x,t)$ exists, we show that it can be represented in terms of the solution of a matrix Riemann-Hilbert problem formulated in the plane of the complex spectral parameter $\x$. This problem has explicit $(x,t)$ dependence, and it has jumps across $\{\x \in \C|\im{\x^4}=0 \}$. The relevant jump matrices are explicitly given in terms of the spectral functions $\{a(\x),b(\x)\},\{A(\x),B(\x)\}$, and $\{\ca({\x}),\cb(\x)\}$, which in turn are defined in terms of the initial data $q_0(x)=q(x,0)$, the boundary data $g_0(t)=q(0,t),g_1(t)=q_x(0,t)$, and another boundary values $f_0(t)=q(L,t),f_1(t)=q_x(L,t)$. The spectral functions are not independent, but related by a compatibility condition, the so-called global relation.
\end{abstract}

\maketitle

\section{Introduction}
\par

The inverse scattering transformation is an important method for solving initial value problems of complete integrable equations,
but it is currently unknown for the case of initial-boundary value problems.
   A new method based on the Riemann-Hilbert
factorization problem  to  solve initial-boundary value problems for nonlinear integrable systems was presented by Forkas  \cite{f1,f2},   and later  it was further developed by several authors \cite{fis,fi1,abmfs,abms}. The Fokas method is based on reducing the initial-boundary value problems to the Riemann-Hilbert problem on the complex plane
of the spectral parameter. But in the analysis of the initial-boundary value problem, we face the problem that for the construction of the associated Riemann-Hilbert problem, more boundary values
are needed  than  a well-posed initial-boundary value problem, since the boundary values  are dependent. In \cite{fi}, Fokas and Its extended the initial-boundary value problem for the nonlinear Schr\"odinger equation  on the half-line to the case of   initial-boundary value on the finite interval.

\par
In this paper, we analyze the Dirichlet initial-boundary value problem for the DNLS equation on a finite interval
\be \label{equ:DNLSequation}
iq_t=-q_{xx}+(r q^2)_x, \quad x\in (0,L),t\in (0,T),
\ee
\be
q(x,0)=q_0(x),x\in (0,L),
\ee
\be
q(0,t)=g_0(t), q(L,t)=f_0(t),t\in (0,T),
\ee
where $r=\pm {\bar q}$,  and $\bar q $ denotes complex conjugate of $q$, the subscripts denote differentiation with respect to the corresponding variables. $L$ and $T$ are positive constants,and $q_0, g_0, f_0$ are smooth functions compatible at $x=t=0$ and at $x=L, t=0$, i.e. $q_0(0)=g_0(0),  q_0(L)=f_0(0)$.

 The DNLS equation (\ref{equ:DNLSequation}) is also called Kaup-Newell equation \cite{kn}.   We just consider $r={\bar q}$,  because the two equations
\[
  iq_t(x,t)q+q_{xx}(x,t)=\pm (|q|^2q)_x
\]
can be transformed into each other by replacing $x \rightarrow {-x}$ \cite{l}.
\par
The DNLS Eq. (\ref{equ:DNLSequation})  has several applications in plasma physics. In plasma physics, it is a model for Alfv$\acute e$n waves propagating
 parallel to the ambient magnetic field, $q$ being the transverse magnetic field perturbation and $x$ and $t$ being space and time
 coordinates, respectively \cite{m}. For more physical  meaning of Eq. (\ref{equ:DNLSequation}), we refer to \cite{k,a,l}, and the references
 therein. Being integrable, Eq. (\ref{equ:DNLSequation}) admits an infinite number of conservation laws and can be analyzed by means of inverse
scattering techniques both in the case of vanishing and nonvanishing boundary conditions \cite{kn, ki}.  A tri-Hamiltonian structure of Eq.(\ref{equ:DNLSequation}) was put forward in \cite{mz}.
The Darboux transformation and soliton solutions   have  been investigated in \cite{ikws,vml, fan,xhw}. Recently, Lenells analyzed its  Riemann-Hilbert problem associated with initial-boundary value problem of the DNLS Eq.(1.1) on the half-line
\cite{l}.

\par
In this paper, we  extend Lenells's  result   to the initial-boundary value on the finite interval (1.1)-(1.3)   following the Fokas and Its idea   \cite{fi}.
We will show that the dependence of Riemann-Hilbert problem associated with initial-boundary values can be  characterized in terms of  spectral functions.
The spectral functions associated with initial and boundary values of a solution for the DNLS equation must satisfy certain global relations.
In the following section 2,  we derive   the spectral analysis of the Lax pair for Eq.(\ref{equ:DNLSequation}).   In section 3,
we investigate the spectral functions $a(\x),b(\x);A(\x),B(\x)$;$\ca({\x}),\cb(\x)$.  Then  Riemann-Hilbert problem associated with
 the initial-boundary value  (1.1)-(1.3)  is  further  presented.

\section{Spectral analysis under the assumption of existence}
The DNLS equation admits the Lax pair formulation \cite{kn,pdl}
\bea \label{Lax:DNLS1}
  v_{1x}+i\x^2v_1=q\x v_2,&&v_{2x}-i\x^2v_2=r\x v_1,\nonumber \\
  iv_{1t}=Av_1+Bv_2,&&iv_{2t}=Cv_1-Av_2,
\eea
where $\x \in \C$ is the spectral parameter, and
\bea
&&A=2\x^4+\x^2rq,B=2i\x^3q-\x q_x+i\x rq^2,\nonumber \\
&&C=2i\x^3r-\x r_x+i\x r^2q.\nonumber
\eea
\par
By introducing
\[
\ba{lcr}
\psi =\left(\ba{c}v_1\\v_2\ea \right),& Q =\left(\ba{lr}0&q\\r&0\ea \right),& \sig_3=\left(\ba{lr}1&0\\0&-1\ea \right), \ea
\]
we can rewrite the Lax pair (\ref{Lax:DNLS1}) in a matrix form

\bea \label{Lax:DNLS2}
  && \psi_x+i\x^2 \sig_3 \psi=\x Q \psi,   \nonumber \\
  && \psi_t+2i\x^4 \sig_3 \psi=(-i\x^2 Q^2 \sig_3+2 \x^3 Q-i\x Q_x \sig_3+\x Q^3) \psi,
\eea

\par
Extending the column vector $\psi$ to a $2\times 2$ matrix and letting
\[
\Psi=\psi e^{i(\x^2 x+2\x^4 t)\sig_3},
\]
we obtain the equivalent Lax pair
\bea \label{Lax:DNLS3}
  && \Psi_x+i\x^2 [\sig_3,\Psi]=\x Q \Psi,   \nonumber \\
  && \Psi_t+2i\x^4 [\sig_3,\Psi]=(-i\x^2 Q^2 \sig_3+2 \x^3 Q-i\x Q_x \sig_3+\x Q^3) \Psi,
\eea
which can be written in full derivative form
\be \label{Lax:DNLS4}
d(e^{i(\x^2 x+2\x^4 t)\hat \sig_3}\Psi(x,t,\x))=e^{i(\x^2 x+2\x^4 t)\hat \sig_3}U(x,t,\x)\Psi,
\ee
where
\be
U=U_1dx+U_2dt=\x Qdx+(-i\x^2 Q^2 \sig_3+2 \x^3 Q-i\x Q_x \sig_3+\x Q^3)dt.
\ee
\par
In order to formulate a Riemann-Hilbert problem for the solution of the inverse spectral problem,we seek solutions of the spectral problem which approach the $2\times 2$ identity matrix  as $\x \rightarrow \infty$. It turns out that solutions of Eq.(\ref{Lax:DNLS4}) do not exhibit this property, hence we use Lenell's method in Ref. \cite{l} to transform the solution $\Psi$ of Eq.(\ref{Lax:DNLS4}) into the desired asymptotic behavior. We write the process of the transformation as follows.

{\bf 2.1 Asymptotic analysis}

\par
Consider a solution of Eq.(\ref{Lax:DNLS4}) of the form
\[
\Psi=D+\frac{\Psi_1}{\x}+\frac{\Psi_2}{\x^2}+\frac{\Psi_3}{\x^3}+O(\frac{1}{\x^4}),\quad \x \rightarrow \infty
\]
where $D,\Psi_1,\Psi_2,\Psi_3$ are independent of $\x$. Substituting the above expansion into the first equation of
(\ref{Lax:DNLS3}),and comparing the same order of $\x$'s frequency, it follows from the $O(\x^2)$ terms that D is a diagonal matrix. Furthermore,one finds the following equations for
 the $O(\x)$ and the diagonal part of the $O(1)$ terms
\[
O(\x):i[\sig_3,\Psi_1]=QD,\quad
i.e. \quad \Psi_1^{(o)}=\frac{i}{2} QD\sig_3,
\]
with $\Psi_1^{(o)}$ being the off-diagonal part of $\Psi_1$,and
\[
O(1):D_x=Q\Psi_1^{(o)},
\]
i.e.
\be \label{2.1.1}
D_x=\frac{i}{2}Q^2\sig_3D.
\ee
\par
On the other hand, substituting the above expansion into the second equation of (\ref{Lax:DNLS3}), one obtains from that
\be\label{2.1.2}
O(\x^3):2i[\sig_3,\Psi_1]=2QD,\quad i.e. \quad \Psi_1^{(o)}=\frac{i}{2}QD\sig_3;
\ee
and
\be \label{2.1.3}
O(\x): 2i[\sig_3,\Psi_3]=-iQ^2 \sig_3 \Psi_1^{(o)}+2Q\Psi_2^{(d)}-iQ_x \sig_3 D+Q^3D,
\ee
i.e.
\be \label{2.1.4}
-iQ^2 \sig_3 \Psi_2^{(d)}+2Q\Psi_3^{(o)}=-\frac{1}{2}Q^3\Psi_1^{(o)}+\frac{1}{2}QQ_xD+\frac{i}{2}Q^4\sig_3D,
\ee
where $\Psi_2^{(d)}$ denotes the diagonal part of $\Psi_2$; and for the diagonal part of the $O(1)$ terms
\[
O(1):D_t=-iQ^2 \sig_3 \Psi_2^{(d)}+2Q\Psi_3^{(o)}-iQ_x\sig_3\Psi_1^{(o)}+Q^3\Psi_1^{(o)},
\]
again,using (\ref{2.1.2}) and (\ref{2.1.4}),we have
\[
D_t=(\frac{3i}{4}Q^4\sig_3+\frac{1}{2}[Q,Q_x])D,
\]
\par
which can be written in terms of $q$ and $r$ as
\be \label{2.1.5}
D_t=(\frac{3i}{4}r^2 q^2+\frac{1}{2}(r_xq-rq_x))\sig_3D.
\ee

{\bf 2.2 Desired Lax pair}
\par
We note that Eq.(\ref{equ:DNLSequation}) admits the conservation law
\[
(\frac{i}{2}rq)_t=(\frac{3i}{4}r^2 q^2+\frac{1}{2}(r_xq-rq_x))_x.
\]
\par
Consequently, the two Eqs.(\ref{2.1.1}) and (\ref{2.1.5}) for $D$ are consistent and are both satisfied if we define
\be \label{D}
D(x,t)=e^{i\int_{(L,0)}^{(x,t)}\Dta \sig_3},
\ee
where $\Dta$ is the closed real-valued one-form
\be \label{Dta}
\Dta(x,t)=\frac{1}{2}dx+(\frac{3}{4}r^2 q^2-\frac{i}{2}(r_xq-rq_x))dt .
\ee
Noting that the integral in (\ref{D}) is independent of the path of integration and the $\Dta$ is independent of $\x$,
then we introduce a new function $\mu$ by
\be \label{mu}
\Psi(x,t,\x)=e^{i\int_{(0,0)}^{(x,t)}\Dta \hat \sig_3}\mu(x,t,\x)D(x,t),
\ee
Thus,we have
\be \label{muinfty}
\mu=\id+O(\frac{1}{\x}),\quad \x \rightarrow \infty,
\ee
and the Lax pair of Eq.(\ref{Lax:DNLS4}) becomes
\be \label{Lax:DNLS5}
d(e^{i(\x^2 x+2\x^4 t)\hat \sig_3}\mu(x,t,\x))=W(x,t,\x),
\ee
where
\[
W(x,t,\x)=e^{i(\x^2 x+2\x^4 t)\hat \sig_3}V(x,t,\x)\mu,
\]

\[
V=V_1dx+V_2dt=e^{-i\int_{(0,0)}^{(x,t)}\Dta \hat \sig_3}(U-i\Dta \sig_3).
\]
Taking into account the definition of $U$ and $\Dta$,we find that
\be \label{V1}
V_1=\left (\ba{ll} -\frac{i}{2}rq&\x qe^{-2i\int_{(0,0)}^{(x,t)}\Dta}\\
                  \x re^{2i\int_{(0,0)}^{(x,t)}\Dta}&\frac{i}{2}rq \ea
                  \right ),
\ee
\be \label{V2}
V_2=\left(\ba{ll} -i\x^2rq-\frac{3i}{4}r^2 q^2-\frac{1}{2}(r_xq-rq_x)&(2\x^3q+i\x q_x+\x q^2r)e^{-2i\int_{(0,0)}^{(x,t)}\Dta}\\
                  (2\x^3r+i\x r_x+\x r^2q)e^{2i\int_{(0,0)}^{(x,t)}\Dta}&i\x^2rq+\frac{3i}{4}r^2 q^2+\frac{1}{2}(r_xq-rq_x)
    \ea \right).
\ee
Then Eq.(\ref{Lax:DNLS5}) for $\mu$ can be written as
\bea \label{Lax:DNLS6}
&&\mu_x+i\x^2[\sig_3,\mu]=V_1\mu,\nonumber \\
&&\mu_t+2i\x^4[\sig_3,\mu]=V_2\mu.
\eea

{\bf 2.3  Eigenfunctions and their relations}
\par
Throughout this section we assume that $q(x,t)$ is sufficiently smooth,in
\[
\Om=\{0<x<L,0<t<T\}
\]
where $T\le \infty$ is a given positive constant;unless otherwise specified,we suppose that $T<\infty$.
\par
Following the idea in Ref.\cite{f}, we define four solutions of Eq.(\ref{Lax:DNLS5}) by
\be \label{muj}
\mu_j(x,t,\x)=\id+\int_{(x_j,t_j)}^{(x,t)}e^{-i(\x^2 x+2\x^4 t)\hat \sig_3}W(y,\tau,\x),\qquad j=1,2,3,4,
\ee
where $(x_1,t_1)=(0,T),(x_2,t_2)=(0,0),(x_3,t_3)=(L,0)$,and $(x_4,t_4)=(L,T)$,see Figure 1.
\par

Since the one-form $W$ is exact, the integral on the righthand side of (\ref{muj}) is independent of the path of integration.We choose the particular contours shown in Figure 2. By splitting the line integrals into integrals parallel to the $t$ and the $x$ axis we find
\be \label{mu1}
\begin{split}
\mu_1(x,t,\x)=\id+\int_0^x e^{i\x^2(y-x)\hat \sig_3}(V_1\mu_1)(y,t,\x)dy &\\
-e^{-i\x^2x\hat \sig_3}\int_t^T e^{2i\x^4(\tau -t)\hat \sig_3}(V_2\mu_1)(0,\tau,\x)d\tau,
\end{split}
\ee
\be \label{mu2}
\begin{split}
\mu_2(x,t,\x)=\id+\int_0^x e^{i\x^2(y-x)\hat \sig_3}(V_1\mu_2)(y,t,\x)dy  &\\
+e^{-i\x^2x\hat \sig_3}\int_0^t e^{2i\x^4(\tau -t)\hat \sig_3}(V_2\mu_2)(0,\tau,\x)d\tau,
\end{split}
\ee
\be \label{mu3}
\begin{split}
\mu_3(x,t,\x)=\id-\int_x^L e^{i\x^2(y-x)\hat \sig_3}(V_1\mu_3)(y,t,\x)dy &\\
+e^{-i\x^2(L-x)\hat \sig_3}\int_0^t e^{2i\x^4(\tau -t)\hat \sig_3}(V_2\mu_3)(L,\tau,\x)d\tau,
\end{split}
\ee
\be \label{mu4}
\begin{split}
\mu_4(x,t,\x)=\id-\int_x^L e^{i\x^2(y-x)\hat \sig_3}(V_1\mu_4)(y,t,\x)dy &\\
-e^{-i\x^2(L-x)\hat \sig_3}\int_t^T e^{2i\x^4(\tau -t)\hat \sig_3}(V_2\mu_4)(L,\tau,\x)d\tau,
\end{split}
\ee
And we note that this choice implies the following inequalities on the contours,
\[
\ba{lcr}
(x_1,t_1) \rightarrow (x,t):&y-x\le 0,&\tau-t\geq 0\\
(x_2,t_2) \rightarrow (x,t):&y-x\le 0,&\tau-t\le 0\\
(x_3,t_3) \rightarrow (x,t):&y-x\geq 0,&\tau-t\le 0\\
(x_4,t_4) \rightarrow (x,t):&y-x\geq 0,&\tau-t\geq 0\\
\ea
\]
We find that the second column of the matrix equation (\ref{muj}) involves $e^{[2i(\x^2(y-x)+2\x^4(\tau-t))]}$ ,
and using the above inequalities it implies that the exponential term of $\mu_j$ is bounded in the following regions of the complex $\x$-plane,
\[
\ba{lr}
(x_1,t_1) \rightarrow (x,t):&\{\im \x^2\le 0\}\cap \{\im \x^4\geq 0\},\\
(x_2,t_2) \rightarrow (x,t):&\{\im \x^2\le 0\}\cap \{\im \x^4\le 0\},\\
(x_3,t_3) \rightarrow (x,t):&\{\im \x^2\geq 0\}\cap \{\im \x^4\le 0\},\\
(x_4,t_4) \rightarrow (x,t):&\{\im \x^2\geq 0\}\cap \{\im \x^4\geq 0\}.\\
\ea
\]
Thus,we have
\be \label{mujdom}
\ba{lccr}
\mu_1=(\mu_1^{(2)},\mu_1^{(3)}),&\mu_2=(\mu_2^{(1)},\mu_2^{(4)}),&\mu_3=(\mu_3^{(3)},\mu_3^{(2)}),&\mu_4=(\mu_4^{(4)},\mu_4^{(1)}),
\ea
\ee
where $\mu_j^{(i)}$ denotes $\mu_j$ is bounded and analytic for $\x \in D_i$,and $D_i=\om_i \cup (-\om_i),-\om_i=\{-\x \in \C |\x \in \om_i\},\om_i=\{\x \in \C|\frac{i-1}{4}\pi <\x <\frac{i}{4}\pi\}$,see Figure 3.
\par
But the functions $\mu_1(0,t,\x)$,$\mu_2(0,t,\x)$,$\mu_3(x,0,\x)$,$\mu_3(L,t,\x)$,$\mu_4(L,t,\x)$ are bounded in larger domains:
\bea \label{mujlardom}
\mu_1(0,t,\x)=(\mu_1^{(24)}(0,t,\x),\mu_1^{(13)}(0,t,\x)),\nonumber\\
\mu_2(0,t,\x)=(\mu_2^{(13)}(0,t,\x),\mu_2^{(24)}(0,t,\x)),&\nonumber \\
\mu_3(x,0,\x)=(\mu_3^{(34)}(x,0,\x),\mu_3^{(12)}(x,0,\x)),\\
\mu_4(L,t,\x)=(\mu_4^{(13)}(L,t,\x),\mu_4^{(24)}(L,t,\x)),&\nonumber \\
\mu_4(L,t,\x)=(\mu_4^{(24)}(L,t,\x),\mu_4^{(13)}(L,t,\x)).&&\nonumber
\eea
By (\ref{muinfty}), it holds that
\be \label{mujinfty}
\mu_j(x,t,\x)=\id +O(\frac{1}{\x}),\quad \x \rightarrow \infty ,j=1,2,3,4.
\ee
The $\mu_j$ are the fundamental eigenfunctions needed for the formulation of a Riemann-Hilbert problem in the complex $\x$-plane.
\par
In order to derive a Riemann-Hilbert problem,we have to compute the jumps across the boundaries of the $D_j$'s. It turns out that the relevant jump matrices can be uniquely defined in terms of three $2\times 2$-matrix valued spectral functions $s(\x),S(\x)$ and $S_L(\x)$ defined as follows.
\par
Assuming $\mu$ and $\tilde \mu$ are the solutions of Eq.(\ref{Lax:DNLS6}),then the two solutions are related by
\be \label{murela}
\mu(x,t,\x)=\tilde \mu(x,t,\x)e^{-i(\x^2 x+2\x^4 t)\hat \sig_3}C_0(\x),
\ee
where $C_0(\x)$ is a $2\times 2$ matrix independent of $x$ and $t$.Let $\psi$ and $\tilde \psi$ be the solutions of Eq.(\ref{Lax:DNLS2}) corresponding to $\mu$ and $\tilde \mu$ according to
\be \label{psimu}
\psi(x,t,\x)=e^{i\int_{(0,0)}^{(x,t)}\Dta \hat \sig_3 }\mu(x,t,\x)D(x,t)e^{-i(\x^2 x+2\x^4 t)\sig_3},
\ee
but we note that there exists a $2\times 2$ matrix $C_1(\x)$ independent of $x$ and $t$ such that
\be \label{psirela}
\psi(x,t,\x)=\tilde \psi(x,t,\x)C_1(\x).
\ee
By using (\ref{murela}),(\ref{psimu}) and (\ref{psirela}),we have
\be \label{Crela}
C_0(\x)=e^{-i\int_{(0,0)}^{(L,0)}\Dta \hat \sig_3}C_1(\x).
\ee
\par
From (\ref{murela}),the functions $\mu_j$ are related by the equations
\be \label{s}
\mu_3(x,t,\x)=\mu_2(x,t,\x)e^{-i(\x^2 x+2\x^4 t)\hat \sig_3}s(\x),
\ee
\be \label{S}
\mu_1(x,t,\x)=\mu_2(x,t,\x)e^{-i(\x^2 x+2\x^4 t)\hat \sig_3}S(\x),
\ee
\be \label{SL}
\mu_4(x,t,\x)=\mu_2(x,t,\x)e^{-i(\x^2 x+2\x^4 t)\hat \sig_3}S_L(\x).
\ee
Evaluating equation (\ref{s}) at $(x,t)=(0,0)$,implies
\be \label{smu3}
s(\x)=\mu_3(0,0,\x).
\ee
Evaluating equation (\ref{S}) at $(x,t)=(0,0)$,gives
\be \label{Smu1}
S(\x)=\mu_1(0,0,\x).
\ee
Evaluating equation (\ref{S}) at $(x,t)=(0,T)$,yields
\be \label{Smu2}
S(\x)=(e^{2i\x^4T\hat \sig_3}\mu_2(0,T,\x))^{-1}.
\ee
Evaluating equation (\ref{SL}) at $(x,t)=(L,0)$,we have
\be \label{SLmu4}
S_L(\x)=\mu_4(L,0,\x).
\ee
Evaluating equation (\ref{SL}) at $(x,t)=(L,T)$,we get
\be \label{SLmu3}
S_L(\x)=(e^{2i\x^4T\hat \sig_3}\mu_3(L,T,\x))^{-1}.
\ee
Eq.(\ref{s}) and (\ref{SL}) imply
\be \label{sSL}
\mu_4(x,t,\x)=\mu_2(x,t,\x)e^{-i(\x^2 x+2\x^4 t)\hat \sig_3}[s(\x)e^{i\x^2 L \hat \sig_3}S_L(\x)],
\ee
which will lead to the global relation.
\par
Hence, the function $s(\x)$ can be obtained from the evaluations at $x=0$ of the function $\mu_3(x,0,\x)$;$S(\x)$ can be obtained from the evaluations at $t=T$ of the function $\mu_2(0,t,\x)$ and $\S_L(\x)$ can be obtained from the evaluations at $t=T$ of the function $\mu_4(L,t,\x)$. And these functions about $\mu_j$ satisfy the linear integral equations
\be \label{mu1x0}
\mu_1(0,t,\x)=\id-\int_t^T e^{2i\x^4(\tau-t)\hat \sig_3}(V_2\mu_1)(0,\tau,\x)d\tau,
\ee
\be \label{mu2x0}
\mu_2(0,t,\x)=\id+\int_0^t e^{2i\x^4(\tau-t)\hat \sig_3}(V_2\mu_2)(0,\tau,\x)d\tau,
\ee
\be \label{mu3t0}
\mu_3(x,0,\x)=\id-\int_x^L e^{i\x^2(y-x)\hat \sig_3}(V_1\mu_3)(y,0,\x)dy,
\ee
\be \label{mu3xL}
\mu_3(L,t,\x)=\id+\int_0^t e^{2i\x^4(\tau-t)\hat \sig_3}(V_2\mu_3)(L,\tau,\x)d\tau,
\ee
\be \label{mu4xL}
\mu_4(L,t,\x)=\id-\int_t^T e^{2i\x^4(\tau-t)\hat \sig_3}(V_2\mu_4)(L,\tau,\x)d\tau.
\ee
By evaluating the equations (\ref{V1}) at $t=0$ and (\ref{V2}) at $x=0$,$x=L$,we find the equations
\be \label{V1t0}
V_1(x,0,\x)=\left(
\ba{lc}
-\frac{i}{2}|q_0|^2&\x q_0e{-i\int_0^x |q_0|^2dy}\\
\x \bar q_0e^{i\int_0^x |q_0|^2dy}&\frac{i}{2}|q_0|^2
\ea
\right),
\ee

\be \label{V2x0}
{\small V_2(0,t,\x)=\left(
\ba{lc}
-i\x^2|g_0|^2-\frac{3i}{4}|g_0|^4-\frac{1}{2}(\bar g_1g_0-\bar g_0g_1)&(2\x^3g_0+i\x g_1+\x g_0|g_0|^2)e^{-2i\int_0^t \Dta_2(0,\tau)d\tau}\\
(2\x^3\bar g_0-i\x \bar g_1+\x \bar g_0|g_0|^2)e^{2i\int_0^t \Dta_2(0,\tau)d\tau}&i\x^2|g_0|^2+\frac{3i}{4}|g_0|^4+\frac{1}{2}(\bar g_1g_0-\bar g_0g_1)
\ea
\right),}
\ee
\be \label{V2xL}
{\small V_2(L,t,\x)=\left(
\ba{lc}
-i\x^2|f_0|^2-\frac{3i}{4}|f_0|^4-\frac{1}{2}(\bar f_1f_0-\bar f_0f_1)&(2\x^3f_0+i\x f_1+\x f_0|f_0|^2)e^{-2i\int_0^t \Dta_2(L,\tau)d\tau}\\
(2\x^3\bar f_0-i\x \bar f_1+\x \bar f_0|f_0|^2)e^{2i\int_0^t \Dta_2(L,\tau)d\tau}&i\x^2|f_0|^2+\frac{3i}{4}|f_0|^4+\frac{1}{2}(\bar f_1f_0-\bar f_0f_1)
\ea
\right).}
\ee
where $q_0(x)=q(x,0)$,$g_0(t)=q(0,t)$,$g_1(t)=q_x(0,t)$,$f_0(t)=q(L,t)$ and $f_1(t)=q_x(L,t)$ are the initial and boundary values of $q(x,t)$,and
\be \label{Dta2x0}
\Dta_2(0,t)=\frac{3}{4}|g_0|^4-\frac{i}{2}(\bar g_1g_0-\bar g_0g_1).
\ee
\be \label{Dta2xL}
\Dta_2(L,t)=\frac{3}{4}|f_0|^4-\frac{i}{2}(\bar f_1f_0-\bar f_0f_1).
\ee
These expressions for $V_1(x,0,\x)$,$V_2(0,t,\x)$ and $V_2(L,t,\x)$ contain only $q_0(x)$,$\{g_0(t),g_1(t)\}$ and $\{f_0(t),f_1(t)\}$, respectively. Therefore, the integral equation (\ref{mu3t0}) determining $s(\x)$ is defined in terms of the initial data $q_0(x)$, the integral equation (\ref{mu2x0}) determining $S(\x)$ is defined in terms of the initial data $\{g_0(t),g_1(t)\}$ and the integral equation (\ref{mu3xL}) determining $S_L(\x)$ is defined in terms of the initial data $\{f_0(t),f_1(t)\}$.
\par
Let us show the function $\mu(x,t,\x)$ satisfy the symmetry relations.
\begin{theorem}\label{symmetry}
For $j=1,2,3,4$,the function $\mu(x,t,\x)=\mu_j(x,t,\x)$ satisfies the symmetry relations \bea \label{sym}
\mu_{11}(x,t,\x)=\ol {\mu_{22}(x,t,\bar \x)},&&\nonumber \\
\mu_{21}(x,t,\x)=\ol {\mu_{12}(x,t,\bar \x)},&&
\eea
as well as
\bea \label{evenodd}
\mu_{11}(x,t,-\x)=\mu_{11}(x,t,\x),&&\nonumber \\
\mu_{12}(x,t,-\x)=-\mu_{12}(x,t,\x),&&\nonumber \\
\mu_{21}(x,t,-\x)=-\mu_{21}(x,t,\x),&&\nonumber \\
\mu_{22}(x,t,-\x)=\mu_{22}(x,t,\x).&&
\eea
\end{theorem}
\begin{proof}
Following the proposition $2.1$'s proof in \cite{l}.
\end{proof}
\par
If $\psi(x,t)$ satisfies (\ref{Lax:DNLS2}),it follows that det($\psi$) is independent of $x$ and $t$.Hence,since det$D(x,t)=1$,the determinant of the function $\mu$ corresponding to $\psi$ according to (\ref{psimu}) is also independent of $x$ and $t$.In particular,for $\mu_j,j=1,2,3,4,$evaluation of det($\mu_j$) at $(x_j,t_j)$ shows that
\be \label{mujdet}
det(\mu_j)=1,j=1,2,3,4,
\ee
In particular,
\[
dets(\x)=detS(\x)=detS_L(\x)=1.
\]
\par
It follows from (\ref{sym}) that
\[
\ba{lc}
s_{11}(\x)=\ol {s_{22}(\bar \x)},&s_{21}(\x)=\ol {s_{12}(\bar \x)},\\
S_{11}(\x)=\ol {S_{22}(\bar \x)},&S_{21}(\x)=\ol {S_{12}(\bar \x)},\\
S_{L11}(\x)=\ol {S_{L22}(\bar \x)},&S_{L21}(\x)=\ol {S_{L12}(\bar \x)},
\ea
\]
so that we use the following notations for $s(\x), S(\x)$ and $S_L(\x)$.
\be \label{sSSL}
\ba{lc}
s(\x)=\left(\ba{lr}\ol{a(\bar \x)}&b(\x)\\
\ol{b(\bar \x)}&a(\x)\ea \right), &S(\x)=\left(\ba{lr}\ol{A(\bar \x)}&B(\x)\\
\ol{B(\bar \x)}&A(\x)\ea \right), \\
S_L(\x)=\left(\ba{lr}\ol{\ca(\bar \x)}&\cb(\x)\\
\ol{\cb(\bar \x)}&\ca(\x)\ea \right)
\ea
\ee
The relations in (\ref{evenodd}) imply that $a(\x),A(\x)$ and $\ca(\x)$ are even functions of $\x$,whereas $b(\x),B(\x)$ and $\cb(\x)$ are odd functions of $\x$,that is,
\be \label{aAeo}
\ba{lr}
a(-\x)=a(\x),&b(-x)=-b(\x)\\
A(-\x)=A(\x),&B(-x)=-B(\x)\\
\ca(-\x)=\ca(\x),&\cb(-x)=-\cb(\x).
\ea
\ee
\par
The definitions of $\mu_3(0,0,\x),\mu_2(0,T,\x),\mu_3(L,0,\x)$ imply
\be \label{smu3x0t0}
s(\x)=\mu_3(0,0,\x)=\id-\int_0^L e^{i\x^2(y-x)\hat \sig_3}(V_1\mu_3)(y,0,\x)dy,
\ee
\be \label{Smu2x0tT}
S^{-1}(\x)=e^{2i\x^4T\hat \sig_3}\mu_2(0,T,\x)=\id+\int_0^T e^{2i\x^4(\tau-t)\hat \sig_3}(V_2\mu_2)(0,\tau,\x)d\tau,
\ee
\be \label{SLmu3xLtT}
S_L^{-1}(\x)=e^{2i\x^4T\hat \sig_3}\mu_3(L,T,\x)=\id+\int_0^T e^{2i\x^4(\tau-t)\hat \sig_3}(V_2\mu_3)(L,\tau,\x)d\tau.
\ee
\par
Equations (\ref{mujlardom}),the determinant conditions (\ref{mujdet}),and the large $\x$ behavior of $\mu_j$ imply the following properties \par
$\ul{a(\x),b(\x)}$
\begin{itemize}
 \item $a(\x),b(\x)$ are defined for $\{\x \in \C|\im \x^2\geq 0\}$ and analytic for $\{\x \in \C|\im \x^2> 0\}$.
 \item $a(\x)\ol{a(\bar \x)}-b(\x)\ol{b(\bar \x)}=1,\x^2 \in \R$.
 \item $a(\x)=1+O(\frac{1}{\x}),b(\x)=O(\frac{1}{\x}),\x \rightarrow \infty ,\im \x^2 \geq 0$.
\end{itemize}
\par
$\ul{A(\x),B(\x)}$
\begin{itemize}
 \item $A(\x),B(\x)$ are defined for $\{\x \in \C|\im \x^4\geq 0\}$ and analytic for $\{\x \in \C|\im \x^4> 0\}$.
 \item $A(\x)\ol{A(\bar \x)}-B(\x)\ol{B(\bar \x)}=1,\x^4 \in \R$. \item $A(\x)=1+O(\frac{1}{\x}),B(\x)=O(\frac{1}{\x}),\x \rightarrow \infty ,\im \x^4 \geq 0$.
\end{itemize}
\par
$\ul{\ca(\x),\cb(\x)}$
\begin{itemize}
 \item $\ca(\x),\cb(\x)$ are defined for $\{\x \in \C|\im \x^4\geq 0\}$ and analytic for $\{\x \in \C|\im \x^4> 0\}$.
 \item $\ca(\x)\ol{\ca(\bar \x)}-\cb(\x)\ol{\cb(\bar \x)}=1,\x^4 \in \R$.
 \item $\ca(\x)=1+O(\frac{1}{\x}),\cb(\x)=O(\frac{1}{\x}),\x \rightarrow \infty ,\im \x^4 \geq 0$.
\end{itemize}

{\bf 2.4  The global relation}
\par
We now show that the spectral functions are not independent but they satisfy an important global relation.
\begin{theorem}
Let the spectral functions $a(\x),b(\x),A(\x),B(\x),\ca(\x),\cb(\x)$,be defined in equations (\ref{sSSL}),where $s(\x),S(\x),S_L(\x)$ are defined by equations (\ref{smu3}),(\ref{Smu2}),(\ref{SLmu3}),and $\mu_2,\mu_3$ are defined by equations (\ref{mu2}),(\ref{mu3}).Then these spectral functions are not independent but they satisfy an important global relation
\be \label{glorela}
(a(\x)\ca{\x}+\ol{b(\bar \x)} e^{2i\x^2L}\cb(\x))B(\x)-(b(\x)\ca(\x)+\ol{a(\bar \x)} e^{2i\x^2L}\cb(\x))A(\x)=e^{4i\x^4T}c^+(\x),
\ee
where $c^+(\x)$ denotes the (12) element of $-\int_0^L (e^{i\x^2y\hat \sig_3})(V_1\mu_4)(y,T,\x)dy$,and $\mu_4$ is defined by (\ref{mu4}).
\end{theorem}
\begin{proof}
We just evaluating equation (\ref{sSL}) at $(x,t)=(0,T)$.
\end{proof}

{\bf 2.5  The jump conditions}
\par
Let $M(x,y,\x)$ be defined by
\be \label{Mdef}
\ba{lr}
M_+=(\frac{\mu_2^{(1)}}{\al(\x)},\mu_4^{(1)}),\x \in D_1,&M_-=(\frac{\mu_1^{(2)}}{d(\x)},\mu_3^{(2)}),\x \in D_2,\\
M_+=(\mu_3^{(3)},\frac{\mu_1^{(3)}}{\ol{d(\bar \x)}}),\x \in D_3,&M_-=(\mu_4^{(4)},\frac{\mu_2^{(4)}}{\ol{\al(\bar \x)}}),\x \in D_4 ,
\ea
\ee
where the scalars $\al(\x)$ and $d(\x)$ are defined below
\be \label{alpha}
\al(\x)=a(\x)\ca(\x)+\ol{b(\bar \x)}e^{2i\x^2L}\cb(\x),
\ee
\be \label{dxi}
d(\x)=a(\x)\ol{A(\bar \x)}-b(\x)\ol{B(\bar \x)}. \ee
\par
These definitions imply
\be \label{Mdet}
detM(x,t,\x)=1,
\ee
and
\be \label{Minf}
M(x,t,\x)=\id+O(\frac{1}{\x}),\quad \x \rightarrow \infty.
\ee
\par
\begin{theorem} \label{RH}
Let $M(x,t,\x)$ be defined by equation (\ref{Mdef}),where $\mu_1(x,t,\x)$,$\mu_2(x,t,\x)$ and $\mu_3(x,t,\x)$,$\mu_4(x,t,\x)$ are defined by equations (\ref{mu1}),(\ref{mu2}) and (\ref{mu3}),(\ref{mu4}),and $q(x,t)$ is a smooth function.Then $M$ satisfies the jump condition
\be \label{Mjump}
M_+(x,t,\x)=M_-(x,t,\x)J(x,t,\x),\quad \x^4\in \R,
\ee
where the $2\times 2$ matrix $J$ is defined by
\be \label{Jdef}
J=\left\{\ba{lc}
J_1,&arg\x^2=0,\\
J_2,&arg\x^2=\frac{\pi}{2},\\
J_3=J_2J_1^{-1}J_4,&arg\x^2=\pi,\\
J_4,&arg\x^2=\frac{3}{2}\pi.
\ea
\right.
\ee
and
\[
J_1=\left(\ba{lc}\frac{1}{\al(\x)\ol{\al(\bar \x)}}&\frac{\beta (\x)}{\ol{\al(\bar \x)}}e^{-2i\tha(\x)}\\
-\frac{\ol{\beta(\bar \x)}}{\al(\x)}e^{2i\tha(\x)}&1
\ea
\right),
\]
\[
J_2=\left(\ba{lc}\frac{a(\x)}{\al(\x)}&\cb(\x)e^{-2i\tha(\x)}e^{2i\x^2L}\\
-\frac{\ol{B(\bar \x)}}{d(\x)\al(\x)}e^{2i\tha(\x)}&\frac{\dta(\x)}{d(\x)}
\ea
\right),
\]
\[
J_4=\left(\ba{lc}\frac{\ol{a(\bar \x)}}{\ol{\al(\bar \x)}}&\frac{B(\x)}{\ol{d(\bar \x)}\ol{\al(\bar \x)}}e^{-2i\tha(\x)}\\
-\ol{\cb(\bar \x)}e^{2i\tha(\x)}e^{-2i\x^2L}&\frac{\ol{\dta(\bar \x)}}{\ol{d(\bar \x)}}
\ea
\right),
\]
\[
\tha(\x)=\x^2x+2\x^4t,
\]
\be \label{albeta}
\al(\x)=a(\x)\ca(\x)+\ol{b(\bar \x)}e^{2i\x^2L}\cb(\x),
\beta(\x)=b(\x)\ca(\x)+\ol{a(\bar \x)}e^{2i\x^2L}\cb(\x),
\ee
\be \label{dtad}
\dta(\x)=\al(\x)\ol{A(\bar \x)}-\beta(\x)\ol{B(\bar \x)},
d(\x)=a(\x)\ol{A(\bar \x)}-b(\x)\ol{B(\bar \x)} .
\ee
\end{theorem}
\begin{proof}
We can following the method of Proposition$2.2$'s proof in \cite{fi}.
\end{proof}
The matrix $M(x,t,\x)$ defined in (\ref{Mdef}) is in general a meromorphic function of $\x$ in $\C \backslash \{\x^4\in \R\}$.The possible poles of $M$ are generated by the zeros of $\al(\x)$,$d(\x)$ and by the complex conjugates of these zeros. Since $a(\x),\ca(\x)$ are even functions and $b(\x),\cb(\x)$ are odd functions,$\al(\x)$ is even function. That means each zero $\x_j$ of $\al(\x)$ is accompanied by another zero at $-\x_j$. Similarly, each zero $\lam_j$ of $d(\x)$ is accompanied by a zero at $-\lam_j$. In particular,both $\al(\x)$ and $d(\x)$ have even number of zeros.
\par
\begin{hypothesis}\label{assumezero}
We assume that
\begin{itemize}
 \item $a(\x)$ has $2\Lam$ simple zeros $\{k_j\}_{j=1}^{2\Lam},2\Lam=2\Lam_1+2\Lam_2$,such that $k_j,j=1,\cdots ,2\Lam_1$,lie in $D_1$ and $\x_j,j=2\Lam_1+1,\cdots ,2n$ lie in $D_2$.
 \item $\al(\x)$ has $2n$ simple zeros $\{\x_j\}_{j=1}^{2n},2n=2n_1+2n_2$,such that $\x_j,j=1,\cdots ,2n$,lie in $D_1$.
 \item $d(\x)$ has $2N$ simple zeros $\{\lam_j\}_{j=1}^{2N},2N=2N_1+2N_2$,such that $\lam_j,j=1,\cdots ,2N$,lie in $D_2$.
 \item None of the zeros of $\al(\x)$ coincides with any of the zeros of $a(\x)$.
 \item None of the zeros of $d(\x)$ coincides with any of the zeros of $a(\x)$.
\end{itemize}
\end{hypothesis}
\par
According to the \ref{assumezero}, we can evaluate the associated residues of $M$.We introduce the notation $[A]_1([A]_2)$ for the first(second) column of a $2\times 2$ matrix $A$ and we also write $\dot a(\x)=\frac{da}{d\x}$. Then we get the following proposition
\begin{proposition}\label{residue}
\be \label{M1resal}
Res_{\x=\x_j}[M(x,t,\x)]_1=c_j^{(1)}e^{2i(\x_j^2x+2\x_j^4t)}[M(x,t,\x_j)]_2,
\ee
\be \label{M2resal}
Res_{\x=\bar \x_j}[M(x,t,\x)]_2=\bar c_j^{(1)}e^{-2i(\bar \x_j^2x+2\bar \x_j^4t)}[M(x,t,\bar \x_j)]_1,
\ee
\be \label{M1resd}
Res_{\x=\lam_j}[M(x,t,\x)]_1=c_j^{(2)}e^{2i(\lam_j^2x+2\lam_j^4t)}[M(x,t,\lam_j)]_2,
\ee
\be \label{M2resd}
Res_{\x=\bar \lam_j}[M(x,t,\x)]_2=\bar c_j^{(2)}e^{-2i(\bar \lam_j^2x+2\bar \lam_j^4t)}[M(x,t,\bar \lam_j)]_1,
\ee
where
\be \label{rescoff}
\ba{c}
c_j^{(1)}=\frac{1}{\dot \al(\x_j)\beta(\x_j)}=\frac{a(\x_j)}{e^{2i\x_j^2L}\cb(\x_j)\dot \al(\x_j)},\\
c_j^{(2)}=\frac{\ol{B(\bar \lam_j)}}{a(\lam_j)\dot d(\lam_j)}.
\ea
\ee
\end{proposition}
\begin{proof}
Following \cite{fi}.
\end{proof}

{\bf 2.6  The inverse problem}
\par
The inverse problem involves reconstructing the potential $q(x,t)$ from the eigenfunctions $\mu_j(x,t,\x),j=1,2,3,4$. We follow the steps of \cite{l}. That means we want to reconstruct the potential $q(x,t)$, then the first step is using any of the four eigenfunctions $\mu_j,j=1,2,3,4$, to compute $m(x,t)$ according to
\[
m(x,t)=\lim_{\x \rightarrow \infty}(\x \mu_j(x,t,\x))_{12}.
\]
The second step is determining $\Dta(x,t)$ by
\[
\ba{l}
rq=4|m|^2,\\
r_xq-rq_x=4(\bar m_xm-m_x\bar m)-32i|m|^4,\\
\Dta=2|m|^2dx-(4|m|^4+2i(\bar m_xm-m_x\bar m))dt.
\ea
\]
finally,$q(x,t)$ is given by
\[
q(x,t)=2im(x,t)e^{2i\int_{(0,0)}^{(x,t)}\Dta}.
\]

\section{The definition of spectral functions and The Riemann-Hilbert Problem}

{\bf 3.1  The definition of spectral functions}
\par
The analysis of section 2 motivates the following definitions for the spectral functions.
\begin{definition}(The spectral functions $a(\x)$ and $b(\x)$)\label{abdef}
Given the smooth function $q_0(x)$, we define the map
\[
\mathbb S:\{q_0(x)\}\rightarrow \{a(\x),b(\x)\}
\]
with
\[
\left(\ba{c}b(\x)\\a(\x)\ea
\right)=[\mu_3(0,\x)]_2,\quad \im \x^2\geq 0.
\]
where $\mu_3(x,\x)$ is the unique solution of the Volterra linear integral equation
\[
\mu_3(x,\x)=\id-\int_x^L e^{i\x^2(y-x)\hat \sig_3}(V_1\mu_3)(y,0,\x)dy,
\]
and $V_1(x,0,\x)$ is given in terms of $q_0(x)$ by (\ref{V1t0}).
\end{definition}
\ul{Properties of $a(\x),b(\x)$}
\begin{itemize}
 \item $a(\x),b(\x)$ are defined for $\{\x \in \C|\im \x^2\geq 0\}$ and analytic for $\{\x \in \C|\im \x^2> 0\}$,
 \item $a(\x)\ol{a(\bar \x)}-b(\x)\ol{b(\bar \x)}=1,\x^2 \in \R$,
 \item $a(\x)=1+O(\frac{1}{\x}),b(\x)=O(\frac{1}{\x}),\x \rightarrow \infty ,\im \x^2 \geq 0$,\\
       in particular,
       \[
       \mbox{$a(\x),b(\x),\ol{a(\bar \x)}e^{2i\x^2L},\ol{b(\bar \x)}e^{2i\x^2L}$  are bounded for  $\im \x^2 \geq 0$}.
       \]
\end{itemize}

\begin{remark}
The definition \ref{abdef} gives rise to the map,
\[
\mathbb S:\{q_0(x)\}\rightarrow \{a(\x),b(\x)\}.
\]
The inverse of this map,
\[
\mathbb Q:\{a(\x),b(\x)\}\rightarrow \{q_0(x)\},
\]
can be defined as follows:
\be \label{Sinverse}
\ba{l}
q_0(x)=2im(x,t)e^{4i\int_0^x|m(y)|^2dy},\\
m(x)=\lim_{\x \rightarrow \infty}(\x M^{(x)}(x,\x))_{12},
\ea
\ee
where $M^{(x)}(x,\x)$ is the unique solution of the following Riemann-Hilbert problem:
\begin{itemize}
\item $M^{(x)}(x,\x)=\left\{\ba{lr}
                          M_-^{(x)}(x,\x)&\im \x^2\le 0,\\
                          M_+^{(x)}(x,\x)&\im \x^2\geq 0.\ea
                          \right.$\\
     is a sectionally meromorphic function.
\item $M_+^{(x)}(x,\x)=M_-^{(x)}(x,\x)J^{(x)}(x,\x),\quad \x^2\in \R$,\\where
\be \label{Mxjump}
J^{(x)}(x,\x)=\left(\ba{lr}\frac{1}{a(\x)\ol{a(\bar \x)}}&\frac{b(\x)}{\ol{a(\bar \x)}}e^{-2i\x^2x}\\
                           -\frac{\ol{b(\bar \x)}}{a(\x)}e^{2i\x^2x}&1
                           \ea
                           \right),\quad \x^2\in \R.
\ee
\item  $M^{(x)}(x,\x)=\id+O(\frac{1}{\x}),\quad \x \rightarrow \infty.$
\item $a(\x)$ has $2\Lam$ simple zeros $\{k_j\}_{j=1}^{2\Lam},2\Lam=2\Lam_1+2\Lam_2$,such that $k_j,j=1,\cdots ,2\Lam_1$,lie in $D_1$,and $k_j,j=2\Lam_1+1,\cdots ,2\Lam$ lie in $D_2$.
\item The first column of $M_+^{(x)}$ has simple poles at $\x=k_j,j=1,\cdots ,2\Lam$,and the second column of $M_-^{(x)}$ has simple poles at $\x=\bar k_j,j=1,\cdots ,2\Lam$.The associated residues are given by
    \be \label{t0RHPresM1}
    Res_{\x=k_j}[M^{(x)}(x,\x)]_1=\frac{1}{\dot a(k_j)b(k_j)}e^{2ik_j^2x}[M^{(x)}(x,k_j)]_2,\quad j=1,\cdots,2\Lam.
    \ee
    \be \label{t0RHPresM2}
    Res_{\x=\bar k_j}[M^{(x)}(x,\x)]_2=\frac{1}{\ol{\dot a(k_j)b(k_j)}}e^{-2i\bar k_j^2x}[M^{(x)}(x,\bar k_j)]_1,\quad j=1,\cdots,2\Lam.
    \ee
\end{itemize}
\end{remark}

\begin{definition}(The spectral functions $A(\x)$ and $B(\x)$)\label{ABdef}
Given the smooth function $g_0(t),g_1(t)$,we define the map
\[
\mathbb S^{(0)}:\{g_0(t),g_1(t)\}\rightarrow \{A(\x),B(\x)\}
\]
with
\[
\left(\ba{c}B(\x)\\A(\x)\ea
\right)=[\mu_1(0,\x)]_2,\quad \im \x^4\geq 0.
\]
where $\mu_1(t,\x)$ is the unique solution of the Volterra linear integral equation
\[
\mu_1(x,\x)=\id-\int_t^T e^{2i\x^4(\tau-t)\hat \sig_3}(V_2\mu_1)(0,\tau,\x)d\tau,
\]
and $V_2(0,t,\x)$ is given in terms of $g_0(t),g_1(t)$ by (\ref{V2x0}).
\end{definition}
\ul{Properties of $A(\x),B(\x)$}
\begin{itemize}
 \item $A(\x),B(\x)$ are defined for $\{\x \in \C|\im \x^4\geq 0\}$ and analytic for $\{\x \in \C|\im \x^4> 0\}$,
 \item $A(\x)\ol{A(\bar \x)}-B(\x)\ol{B(\bar \x)}=1,\x^4 \in \R$,
 \item $A(\x)=1+O(\frac{1}{\x}),B(\x)=O(\frac{1}{\x}),\x \rightarrow \infty ,\im \x^4 \geq 0$,in particular,
 \[
 \mbox{$A(\x),B(\x)$ are bounded for $\x \in D_1\cup D_2$.}
 \]
\end{itemize}
\begin{remark}\label{ABinv}
The definition \ref{ABdef} gives rise to the map,
\[
\mathbb S^{(0)}:\{g_0(x),g_1(x)\}\rightarrow \{A(\x),B(\x)\}
\]
The inverse of this map,
\[
\mathbb Q^{(0)}:\{A(\x),B(\x)\}\rightarrow \{g_0(x),g_1(x)\},
\]
can be defined as follows:
\be \label{S0inverse}
\ba{l}
g_0(t)=2im_{12}^{(1)}(t)e^{2i\int_0^t\Dta_2(\tau)d\tau},\\
g_1(t)=(4m_{12}^{(3)}(t)+|g_0(t)|^2m_{12}^{(1)}(t))e^{2i\int_0^t\Dta_2(\tau)d\tau}+ig_0(t)(2m_{22}^{(2)}(t)+|g_0(t)|^2),
\ea
\ee
where \[
\Dta_2(t)=4|m_{12}^{(1)}|^4+8(\re [m_{12}^{(1)}\bar m_{12}^{(3)}]-|m_{12}^{(1)}|^2\re [m_{22}^{(2)}]).
\]
The functions $m^{(1)}(t),m^{(2)}(t),m^{(3)}(t)$ are determined by the asymptotic expansion
\[
M^{(t)}(t,\x)=\id+\frac{m^{(1)}(t)}{\x}+\frac{m^{(2)}(t)}{\x^2}+\frac{m^{(3)}(t)}{\x^3}+O(\frac{1}{\x^4}),\quad \x \rightarrow \infty,
\]
where $M^{(t)}(t,\x)$ is the unique solution of the following Riemann-Hilbert problem:
\begin{itemize}
\item $M^{(t)}(t,\x)=\left\{\ba{lr}
                          M_-^{(t)}(t,\x)&\im \x^4\le 0,\\
                          M_+^{(t)}(t,\x)&\im \x^4\geq 0.\ea
                          \right.$\\
     is a sectionally meromorphic function.
\item $M_+^{(t)}(t,\x)=M_-^{(t)}(t,\x)J^{(t,0)}(t,\x),\quad \x^4\in \R$,\\where
\be \label{Mxjump}
J^{(t,0)}(t,\x)=\left(\ba{lr}\frac{1}{A(\x)\ol{A(\bar \x)}}&\frac{B(\x)}{\ol{A(\bar \x)}}e^{-4i\x^4t}\\
                           -\frac{\ol{B(\bar \x)}}{A(\x)}e^{4i\x^4t}&1
                           \ea
                           \right),\quad \x^4\in \R.
\ee
\item  $M^{(t)}(t,\x)=\id+O(\frac{1}{\x}),\quad \x \rightarrow \infty.$
\item $A(\x)$ has $2A$ simple zeros $\{K_j\}_{j=1}^{2A},2A=2A_1+2A_2$,such that $K_j,j=1,\cdots ,2A_1$,lie in $D_1$,and $K_j,j=2A_1+1,\cdots ,2A$ lie in $D_3$.
\item The first column of $M_+^{(t)}$ has simple poles at $\x=K_j,j=1,\cdots ,2A$,and the second column of $M_-^{(t)}$ has simple poles at $\x=\bar K_j,j=1,\cdots ,2A$.The associated residues are given by
    \be \label{x0RHPresM1}
    Res_{\x=K_j}[M^{(t)}(t,\x)]_1=\frac{1}{\dot A(K_j)B(K_j)}e^{4iK_j^4t}[M^{(t)}(t,K_j)]_2,\quad j=1,\cdots,2A.
    \ee
    \be \label{x0RHPresM2}
    Res_{\x=\bar K_j}[M^{(t)}(t,\x)]_2=\frac{1}{\ol{\dot A(K_j)B(K_j)}}e^{-4i\bar K_j^4t}[M^{(t)}(t,\bar K_j)]_1,\quad j=1,\cdots,2A.
    \ee
\end{itemize}
\end{remark}

\begin{definition}(The spectral functions $\ca(\x)$ and $\cb(\x)$)\label{cacbdef}
Given the smooth function $f_0(t),f_1(t)$,we define the map
\[
\mathbb S^{(L)}:\{f_0(t),f_1(t)\}\rightarrow \{\ca(\x),\cb(\x)\}
\]
with
\[
\left(\ba{c}\cb(\x)\\\ca(\x)\ea
\right)=[\mu_4(0,\x)]_2,\quad \im \x^4\geq 0.
\]
where $\mu_4(t,\x)$ is the unique solution of the Volterra linear integral equation
\[
\mu_4(x,\x)=\id-\int_t^T e^{2i\x^4(\tau-t)\hat \sig_3}(V_2\mu_4)(L,\tau,\x)d\tau,
\]
and $V_2(L,t,\x)$ is given in terms of $f_0(t),f_1(t)$ by (\ref{V2xL}).
\end{definition}
\ul{Properties of $\ca(\x),\cb(\x)$}
\begin{itemize}
 \item $\ca(\x),\cb(\x)$ are defined for $\{\x \in \C|\im \x^4\geq 0\}$ and analytic for $\{\x \in \C|\im \x^4> 0\}$,
 \item $\ca(\x)\ol{\ca(\bar \x)}-\cb(\x)\ol{\cb(\bar \x)}=1,\x^4 \in \R$,
 \item $\ca(\x)=1+O(\frac{1}{\x}),\cb(\x)=O(\frac{1}{\x}),\x \rightarrow \infty ,\im \x^4 \geq 0$.
\end{itemize}
\begin{remark}  \label{cacb}
The definition \ref{cacbdef} gives rise to the map,
\[
\mathbb S^{(L)}:\{f_0(x),f_1(x)\}\rightarrow \{\ca(\x),\cb(\x)\}.
\]
The inverse of this map,
\[
\mathbb Q^{(L)}:\{\ca(\x),\cb(\x)\}\rightarrow \{f_0(x),f_1(x)\},
\]
can be defined as follows
\be \label{SLinverse}
\ba{l}
f_0(t)=2im_{12}^{(1)}(t)e^{2i\int_0^t\Dta_2^L(\tau)d\tau},\\
f_1(t)=(4m_{12}^{(3)}(t)+|g_0(t)|^2m_{12}^{(1)}(t))e^{2i\int_0^t\Dta_2^L(\tau)d\tau}+ig_0(t)(2m_{22}^{(2)}(t)+|g_0(t)|^2),
\ea
\ee
where \[
\Dta_2^L(t)=4|m_{12}^{(1)}|^4+8(\re [m_{12}^{(1)}\bar m_{12}^{(3)}]-|m_{12}^{(1)}|^2\re [m_{22}^{(2)}]),
\]
and the functions $m^{(1)}(t),m^{(2)}(t),m^{(3)}(t)$ are determined by the asymptotic expansion
\[
M^{(t)}(t,\x)=\id+\frac{m^{(1)}(t)}{\x}+\frac{m^{(2)}(t)}{\x^2}+\frac{m^{(3)}(t)}{\x^3}+O(\frac{1}{\x^4}),\quad \x \rightarrow \infty,
\]
where $M^{(t)}(t,\x)$ is the unique solution of the following Riemann-Hilbert problem:
\begin{itemize}
\item $M^{(t)}(t,\x)=\left\{\ba{lr}
                          M_-^{(t)}(t,\x)&\im \x^4\le 0,\\
                          M_+^{(t)}(t,\x)&\im \x^4\geq 0.\ea
                          \right.$\\
     is a sectionally meromorphic function.
\item $M_+^{(t)}(t,\x)=M_-^{(t)}(t,\x)J^{(t,L)}(t,\x),\quad \x^4\in \R$,\\where
\be \label{MxLjump}
J^{(t,L)}(t,\x)=\left(\ba{lr}\frac{1}{\ca(\x)\ol{\ca(\bar \x)}}&\frac{\cb(\x)}{\ol{\ca(\bar \x)}}e^{-4i\x^4t}\\
                           -\frac{\ol{\cb(\bar \x)}}{\ca(\x)}e^{4i\x^4t}&1
                           \ea
                           \right),\quad \x^4\in \R.
\ee
\item  $M^{(t)}(t,\x)=\id+O(\frac{1}{\x}),\quad \x \rightarrow \infty.$
\item $\ca(\x)$ has $2\ca$ simple zeros $\{\cK_j\}_{j=1}^{2\ca},2\ca=2\ca_1+2\ca_2$,such that $\cK_j,j=1,\cdots ,2\ca_1$,lie in $D_1$,and $\cK_j,j=2\ca_1+1,\cdots ,2\ca$ lie in $D_3$.
\item The first column of $M_+^{(t)}$ has simple poles at $\x=\cK_j,j=1,\cdots ,2\ca$,and the second column of $M_-^{(t)}$ has simple poles at $\x=\bar \cK_j,j=1,\cdots ,2\ca$.The associated residues are given by
    \be \label{xLRHPresM1}
    Res_{\x=\cK_j}[M^{(t)}(t,\x)]_1=\frac{1}{\dot \ca(\cK_j)\cb(\cK_j)}e^{4i\cK_j^4t}[M^{(t)}(t,\cK_j)]_2,\quad j=1,\cdots,2\ca.
    \ee
    \be \label{xLRHPresM2}
    Res_{\x=\bar \cK_j}[M^{(t)}(t,\x)]_2=\frac{1}{\ol{\dot \ca(\cK_j)\cb(\cK_j)}}e^{-4i\bar \cK_j^4t}[M^{(t)}(t,\bar \cK_j)]_1,\quad j=1,\cdots,2\ca.
    \ee
\end{itemize}
\end{remark}

\begin{definition}(An admissible set).
Given the smooth function $q_0(x)$ define $a(\x),b(\x)$ according to definition \ref{abdef}.Suppose that there exist smooth functions
\\ $g_0(t),g_1(t),f_0(t),f_1(t)$, such that \begin{itemize}
 \item The associated $A(\x),B(\x),\ca(\x),\cb(\x)$,defined according to definition \ref{ABdef} and \ref{cacbdef},satisfy the relation
  \be \label{globalrela}
  \ba{l}
  (a(\x)\ca(\x)+\ol{b(\bar \x)} e^{2i\x^2L}\cb(\x))B(\x)-(b(\x)\ca(\x)+\ol{a(\bar \x)} e^{2i\x^2L}\cb(\x))A(\x)\\
  =e^{4i\x^4T}c^+(\x), \ \x \in \C,
  \ea
  \ee
 where $c^+(\x)$ is an entire function,which is bounded for $\im \x^2\geq 0$ and $c^+(\x)=O(\frac{1}{\x})$,as $\x \rightarrow \infty$.
 \item $g_0(0)=q_0(0),g_1(0)=q'_0(0),f_0(0)=q_0(L),f_1(0)=q'_0(L)$.
\end{itemize}
Then we call the functions $g_0(t),g_1(t),f_0(t),f_1(t)$,an admissible set of functions with respect to $q_0(x)$.
\end{definition}

{\bf 3.2  The Riemann-Hilbert problem}
\par
\begin{theorem}\label{RHP}
Let $q_0(x)$ be a smooth function.Suppose that the set of functions $g_0(t),g_1(t),f_0(t),f_1(t)$,are admissible with respect to $q_0(x)$. Define the spectral functions $a(\x),b(\x),A(\x),B(\x),\ca(\x),\cb(\x)$,in terms of $q_0(x),g_0(t),g_1(t)$,
$f_0(t),f_1(t)$. According to the (\ref{assumezero}). Define $M(x,t,\x)$ as the solution of the following $2\times 2$ matrix Riemann-Hilbert problem
\begin{itemize}
 \item $M$ is sectionally meromorphic in $\C \backslash \{\x^4\in \R\}$,and has unit determinant.
 \item $M$ satisfies the jump condition
        \[
        M_+(x,t,\x)=M_-(x,t,\x)J(x,t,\x),\quad \x^4\in \R.
        \]
        where $M$ is $M_+$ for $\im \x^4\geq 0$,$M$ is $M_-$ for $\im \x^4\le 0$,and $J$ is defined in terms of $a,b,A,B,\ca,\cb$ by Eq.(\ref{Jdef}).
 \item \[
        M(x,t,\x)=\id+O(\frac{1}{\x}),\x \rightarrow \infty.
       \]
 \item Residue conditions (\ref{M1resal})-(\ref{M2resd}).

\end{itemize}
Then $M(x,t,\x)$ exists and is unique.
\par
Define $q(x,t)$ in terms of $M(x,t,\x)$ by
\be \label{qsoluM}
\ba{l}
q(x,t)=2im(x,t)e^{2i\int_{(0,0)}^{(x,t)}\Dta},\\
m(x,t)=\lim_{\x \rightarrow \infty}(\x \mu_j(x,t,\x))_{12},\\
\Dta=2|m|^2dx-(4|m|^4+2i(\bar m_xm-m_x\bar m))dt.
\ea
\ee
Then $q(x,t)$ solves the DNLS equation (\ref{equ:DNLSequation}) with
\[
q(x,0)=q_0(x),q(0,t)=g_0(t),q_x(0,t)=g_1(t),q(L,t)=f_0(t),q_x(L,t)=f_1(t).
\]
\end{theorem}
\begin{proof}
\par
If $\al(\x)$ and $d(\x)$ have no zeros for $\x \in D_1$ and for $\x \in D_2$ respectively, then the function $M(x,t,\x)$ satisfies a non-singular Riemann-Hilbert problem. Using the fact that the jump matrix $J$ satisfies appropriate symmetry conditions it is possible to show that this problem has a unique global solution \cite{aggh}.The case that $\al(\x)$ and $d(\x)$ have a finite number of zeros can be mapped to the case of no zeros supplemented by an algebraic system of equations which is always uniquely solvable \cite{aggh}.
\par
\mbox{\bf Proof that $q(x,t)$ satisfies the DNLS equation}
\par
Using arguments of the dressing method\cite{zs}, it can be verified directly that if $M(x,t,\x)$ is defined as the unique solution of the above Riemann-Hilbert problem,and if $q(x,t)$ is defined in terms of $M$ by equation(\ref{qsoluM}), then $q$ and $M$ satisfy both parts of the Lax pair, hence $q$ solves the DNLS equation.
\par
\mbox{\bf Proof that $q(x,0)=q_0(x)$}.
\par
Evaluating the equation(\ref{Jdef}) at $t=0$, we  can  divide the jump matrix into product of $2\times 2$ matrix.
By changing  the problem into the half-line Riemann-Hilbert problem,  we can prove $q(x,0)=q_0(x)$ in a similar way to Ref.\cite{l}.
\par
Evaluating the equation(\ref{Jdef}) at $t=0$:
\be \label{J1t0}
J_1=\left(\ba{lc}\frac{1}{\al(\x)\ol{\al(\bar \x)}}&\frac{\beta (\x)}{\ol{\al(\bar \x)}}e^{-2i\x^2x}\\
-\frac{\ol{\beta(\bar \x)}}{\al(\x)}e^{2i\x^2x}&1
\ea
\right),
\ee
\be \label{J2t0}
J_2=\left(\ba{lc}\frac{a(\x)}{\al(\x)}&\cb(\x)e^{-2i\x^2x}e^{2i\x^2L}\\
-\frac{\ol{B(\bar \x)}}{d(\x)\al(\x)}e^{2i\x^2x}&\frac{\dta(\x)}{d(\x)}
\ea
\right),
\ee
\be \label{J4t0}
J_4=\left(\ba{lc}\frac{\ol{a(\bar \x)}}{\ol{\al(\bar \x)}}&\frac{B(\x)}{\ol{d(\bar \x)}\ol{\al(\bar \x)}}e^{-2i\x^2x}\\
-\ol{\cb(\bar \x)}e^{2i\x^2x}e^{-2i\x^2L}&\frac{\ol{\dta(\bar \x)}}{\ol{d(\bar \x)}}
\ea
\right).
\ee
And then we introduce some $2\times 2$ matrix
\be \label{Jhatinfty}
\ba{lc}
J_1^{(\infty)}={\left(\ba{lc}\frac{1}{a(\x)\ol{a(\bar \x)}}&\frac{b(\x)}{\ol{a(\bar \x)}}e^{-2i\x^2x}\\-\frac{\ol{b(\bar \x)}}{a(\x)}e^{2i\x^2x}&1 \ea \right)},\\
J_2^{(\infty)}={\left(\ba{lc}1&0 \\-\frac{\ol{B(\bar \x)}}{a(\x)d(\x)}e^{2i\x^2x}&1 \ea \right)},\\
J_4^{(\infty)}={\left(\ba{lc}1&\frac{\ol{B(\bar \x)}}{\ol{a(\bar \x)}\ol{d(\bar \x)}}e^{-2i\x^2x}\\0&1 \ea \right)},\\
\hat J_1={\left(\ba{lc}\frac{a(\x)}{\al(\x)}&\cb(\x)e^{2i\x^2L}e^{-2i\x^2x}\\0&\frac{\al(\x)}{a(\x)}\ea \right)},\\
\hat J_4={\left(\ba{lc}\frac{\ol{a(\bar \x)}}{\ol{\al(\bar \x)}}&0\\-\ol{\cb(\bar \x)}e^{-2i\x^2L}e^{2i\x^2x}&\frac{\ol{\al(\bar \x)}}{\ol{a(\bar \x)}}\ea \right)}.&
\ea
\ee
Thus we can verify that:
\be \label{JandJhatrela}
\ba{lc}
J_1(x,0,\x)=\hat J_4J_1^{(\infty)}\hat J_1,&J_2(x,0,\x)=J_2^{(\infty)}\hat J_1,\\
J_3(x,0,\x)=J_2^{(\infty)}(J_1^{(\infty)})^{-1}J_4^{(\infty)}&J_4(x,0,\x)=\hat J_4J_4^{(\infty).}
\ea
\ee
Let $M^{(1)}(x,t,\x),M^{(2)}(x,t,\x),M^{(3)}(x,t,\x),M^{(4)}(x,t,\x)$,denote $M(x,t,\x)$ for $\x \in D_1,\x \in D_2,\x \in D_3,\x \in D_4$.Then the jump condition(\ref{Mjump}) becomes
\be \label{Mjumpt0}
M^{(1)}=M^{(4)}J_1,M^{(1)}=M^{(2)}J_2,M^{(3)}=M^{(2)}J_3,M^{(3)}=M^{(4)}J_4.
\ee
Using equations(\ref{JandJhatrela}),  we find
\be \label{Mjumphatinfty}
\ba{lc}
M^{(1)}(x,0,\x)=M^{(4)}(x,0,\x)\hat J_4J_1^{(\infty)}\hat J_1,\\
M^{(1)}(x,0,\x)=M^{(2)}(x,0,\x)J_2^{(\infty)}\hat J_1,\\
M^{(3)}(x,0,\x)=M^{(2)}(x,0,\x)J_2^{(\infty)}(J_1^{(\infty)})^{-1}J_4^{(\infty)},\\
M^{(3)}(x,0,\x)=M^{(4)}(x,0,\x)\hat J_4J_4^{(\infty).}
\ea
\ee
Defining $M_j^{(\infty)},j=1,2,3,4,$ by
\be \label{Minftydef}
\ba{lc}
M_1^{(\infty)}=M^{(1)}(x,0,\x)(\hat J_1)^{-1},&M_2^{(\infty)}=M^{(2)}(x,0,\x),\\
M_3^{(\infty)}=M^{(3)}(x,0,\x),&M_4^{(\infty)}=M^{(4)}(x,0,\x)\hat J_4.
\ea
\ee
then we find that the sectionally holomorphic function $M^{(\infty)}(x,\x)$ satisfies the jump conditions
\be \label{Minfjump}
\ba{lc}
M_1^{(\infty)}=M_4^{(\infty)}J_1^{(\infty)}, &M_1^{(\infty)}=M_2^{(\infty)}J_2^{(\infty)}, \\
M_3^{(\infty)}=M_2^{(\infty)}J_3^{(\infty)},&
M_3^{(\infty)}=M_4^{(\infty)}J_4^{(\infty)}.
\ea
\ee
These conditions are precisely the jump conditions satisfied by the unique solution of the Riemann-Hilbert problem associated with DNLS for $0<x<\infty,0<t<T$\cite{l}.Also $detM^{(\infty)}=1$ and $M^{(\infty)}=\id+O(\frac{1}{\x}),\x \rightarrow \infty.$Moreover,by a straightforward calculation one can verify that the associated residue conditions change into the proper residue conditions\cite{l}.
Therefore,$M^{(\infty)}(x,\x)$ satisfies the same Riemann-Hilbert problem as the Riemann-Hilbert problem associated with the half-line evaluated at $t=0$.Hence,$q(x,0)=q_0(x).$
\par
\mbox{\bf Proof that $q(0,t)=g_0(t),q_x(0,t)=g_1(t)$}
\par
Let $M^{(t,0)}(t,\x)$ be defined by
\be \label{Mt0def}
  M^{(t,0)}(t,\x)=M(0,t,\x)G(t,\x),
\ee
where $G$ is given by $G^{(1)},G^{(2)},v,G^{(3)},G^{(4)}$,for $\x \in D_1,\x \in D_2,\x \in D_3,\x \in D_4$.Suppose we can find matrices $G^{(j)}$ which are holomorphic,tend to $\id$ as $\x \rightarrow \infty$,and satisfy
\bea \label{Gcond}
   &J_2(0,t,\x)G^{(1)}(t,\x)=G^{(2)}(t,\x)J^{(t,0)}(\x),&   \nn \\
   &J_1(0,t,\x)G^{(1)}(t,\x)=G^{(4)}(t,\x)J^{(t,0)}(\x), &  \nn \\
   &J_4(0,t,\x)G^{(3)}(t,\x)=G^{(4)}(t,\x)J^{(t,0)}(\x), &
\eea
where $J^{(t,0)}(\x)$ is defined in(\ref{Mxjump}). Then equation(\ref{Gcond}) yield $J_3(0,t,\x)G^{(3)}(t,\x)=G^{(2)}(t,\x)J^{(t,0)}(\x)$, and equations(\ref{Mjumpt0}), (\ref{Mt0def}) imply that $M^{(t,0)}(t,\x)$ satisfies the Riemann-Hilbert problem defined in Remark \ref{ABinv}. Then the Remark \ref{ABinv} implies the desired result.
\par
We will show that such $G^{(j)}$ matrices exist and can be written  as
\be \label{Gdef}
 \begin{split}
  G^{(1)}(t,\x)=\left( \ba{lc} \frac{\al(\x)}{A(\x)}&c^+(\x)e^{4i\x^4(T-t)}\\0&\frac{A(\x)}{\al(\x)}\ea \right), \\
  G^{(2)}(t,\x)=\left(\ba{lc} d(\x)&-\frac{b(\x)}{\ol{A(\bar \x)}}e^{-4i\x^4t}\\0&\frac{1}{d(\x)}\ea \right),\\
  G^{(3)}(t,\x)=\left(\ba{lc} \frac{1}{\ol{d(\bar \x)}}&0\\-\frac{\ol{b(\bar \x)}}{A(\x)}e^{4i\x^4t}&\ol{d(\bar \x)}\ea \right),\\
  G^{(4)}(t,\x)=\left(\ba{lc}\frac{\ol{A(\bar \x)}}{\ol{\al(\bar \x)}}&0\\\ol{c^+(\bar \x)}e^{-4i\x^4(T-t)}&\frac{\ol{\al{\bar \x}}}{\ol{A(\bar \x)}}\ea \right).
 \end{split}
\ee
We use straight forward calculation to verify these $G^{(j)}$ matrices satisfy the conditions(\ref{Gcond}).
Similar to the proof of the equation $q(x,0)=q_0(x)$, it can be verified that the transformation (\ref{Mt0def}) replaces the residue conditions (\ref{M1resal})
-(\ref{M2resd}) by the residue conditions of Remark \ref{ABinv}.

\par
\mbox{\bf Proof that $q(L,t)=f_0(t),q_x(L,t)=f_1(t)$}
\par
Following the arguments similar to the prove above we must show that the matrices $F^{(j)}(t,\x)$ such that
\bea \label{FtLcon}
&J_2(L,t,\x)F^{(1)}(t,\x)=F^{(2)}(t,\x)J^{(t,L)}(\x),&   \nn \\
&J_1(L,t,\x)F^{(1)}(t,\x)=F^{(4)}(t,\x)J^{(t,L)}(\x), &  \nn \\
&J_4(L,t,\x)F^{(3)}(t,\x)=F^{(4)}(t,\x)J^{(t,L)}(\x). &
\eea
We will show that such $F^{(j)}$ matrices are
\be \label{Fdef}
\begin{split}
 F^{(1)}(t,\x)=\left(\ba{lc}-1&0\\-\frac{\ol{b(\bar \x)}e^{4i\x^4t+2i\x^2L}}{\al(\x)\ca(\x)}&-1\ea \right),\\
 F^{(2)}(t,\x)=\left(\ba{lc}-\ol{\ca(\bar \x)}&0\\ \frac{\ol{c^+(\bar \x)}e^{-4i\x^4(T-t)+2i\x^2L}}{d(\x)}&-\frac{1}{\ol{\ca(\bar \x)}}\ea \right),\\
 F^{(3)}(t,\x)=\left(\ba{lc}-\frac{1}{\ca(\x)}&\frac{c^+(\x)e^{4i\x^4(T-t)-2i\x^2L}}{\ol{d(\bar \x)}}\\0&-\ca(\x)\ea \right),\\
 F^{(4)}(t,\x)= \left(\ba{lc}-1&-\frac{b(\x)e^{-4i\x^4t-2i\x^2L}}{\ol{\al(\bar \x)\ca(\bar \x)}}\\0&-1\ea \right).
\end{split}
\ee
Similar to the previous case, the transformation
\be \label{MxLrel}
M(L,t,\x)\rightarrow M^{(t,L)}(t,\x)=M(L,t,\x)F(t,\x)
\ee
maps the Riemann-Hilbert problem of Theorem \ref{RHP} to the Riemann-Hilbert problem of Remark \ref{cacb}.
\end{proof}

\begin{remark}
It is well-known that  there are  three kinds of celebrated DNLS equations,  including   Kaup-Newell equation ( i.e Eq.(\ref{equ:DNLSequation})),
Chen-Lee-Liu equation \cite{chen2}
\be \label{CLL}
iq_t+q_{xx}+i|q|^2q_x=0,
\ee
and Gerdjikov-Ivanov equation \cite {kundu}
\be \label{GI}
iq_t+q_{xx}-iq^2\bar{q}_x+\frac{1}{2}|q|^4\bar{q}=0.
\ee
It has been  found that they may be transformed into each
other by  gauge transformations \cite{kundu, fan2}.  We can show that,   for  the decaying and smooth potential  $q$ ,
the Jost functions associated with   spectral problems  for  these three kinds of DNLS equations have the same asymptotic
behavior as $|x|\rightarrow \infty$,   so   the Chen-Lee-Liu equation (\ref{CLL})  and  Gerdjikov-Ivanov equation (\ref{GI})
 have  the same type Riemann-Hilbert problem  \ref{RHP} with the DNLS  equation (1.1).

\end{remark}

{\bf Acknowledgments.}
This work is  supported by  the National Science Foundation of China (No.10971031)
and Shanghai Shuguang Tracking Project (No.08GG01).

\end{document}